\documentclass[runningheads]{llncs}
\usepackage{graphicx}
\usepackage{tikz-cd}
\usetikzlibrary{decorations.pathreplacing}
\usepackage{hyperref}

\usepackage{enumitem}

\newcommand{\abs}[1]{\left| #1 \right|}
\newcommand{\bigo}{\mathcal{O}}

\newcommand{\PP}{\textsc{Perfect Phylogeny}}
\newcommand{\TCG}{\textsc{Triangulating Colored Graphs}}
\newcommand{\TMG}{\textsc{Triangulating Multicolored Graphs}}
\newcommand{\MIS}{\textsc{Multicolor Independent Set}}
\newcommand{\TCMIS}{\textsc{Tree-chained Multicolor Independent Set}}

\newcommand{\prblm}[5]
{
    \noindent\fbox{\begin{minipage}{\dimexpr\linewidth-2\fboxsep-2\fboxrule}
        \textsc{#1} (#2)\\
        \textbf{Input:} #3\\
        \textbf{Parameter:} #4\\
        \textbf{Question:} #5
    \end{minipage}}
}

\begin{document}
\title{On the Parameterized Complexity of the Perfect Phylogeny Problem}
%
%
\author{Jorke M. de Vlas}
\authorrunning{J.M. de Vlas}
%
\institute{Utrecht University, the Netherlands; Link\"oping Universitet, Sweden}
\maketitle              
\begin{abstract}
This paper categorizes the parameterized complexity of the algorithmic problems \PP{} and \TCG{} when parameterized by the number of genes and colors, respectively. We show that they are complete for the parameterized complexity class XALP using a reduction from \TCMIS{} and a proof of membership. We introduce the problem \TMG{} as a stepping stone and prove XALP-completeness for this problem as well. We also show that, assuming the Exponential Time Hypothesis, there exists no algorithm that solves any of these problems in time $f(k)n^{o(k)}$, where $n$ is the input size, $k$ the parameter, and $f$ any computable function.


\subsubsection*{Acknowledgements.}
This paper was written as a master thesis at Utrecht University. I wish to thank my supervisors Hans L. Bodlaender and Carla Groenland for the discussions and guidance.
\end{abstract}


\section{Introduction}
\label{sec:intro}
A phylogeny is a tree that describes the evolution history of a set $S$ of species. Every vertex corresponds to a species: leafs correspond to species from $S$, and internal vertices correspond to hypothetical ancestral species. Species are characterized by their gene-variants, and the quality of a phylogeny is determined by how well it represents those variants. In particular, a phylogeny is \emph{perfect} if each gene-variant was introduced at exactly one point in the tree. That is, the subset of vertices that contain the variant is connected. \PP{} is the algorithmic problem of determining the existence of a perfect evolutionary tree. It has large implications on determining the evolutionary history of genetic sequences and is therefore of major importance. This application is not limited to biology: it can also be used to determine the history of languages or cultures.

The concept of phylogenies as an algorithmic problem has been well researched since the 60s. The first formal definition of \PP{} was given by Estabrook~\cite{EJM_76}. In 1974, Buneman showed that the problem can be reduced to the more combinatorial \TCG{}~\cite{B_74} which by itself has also become an important, well-studied problem. An inverse reduction, and thus equivalence, was given by Kannan and Warnow~\cite{KW_90}. In 1992, Bodlaender et al showed that \PP{} is NP-complete~\cite{BFW_92}.

After Downey and Fellows introduced parameterized complexity~\cite{DF_95}, people have tried to determine the complexity of \PP{} when seen as a parameterized problem. There are two main ways to parameterize the problem: either by using the number of genes or by using the maximum number of variants for each gene. In the second case, the problem becomes FPT~\cite{KW_97}. In the first case, the parameterized complexity was unknown. There are some partial results: On one hand, it was shown that the problem is $W[t]$-hard for every $t$~\cite{BFH+_00}. On the other hand, there exists an algorithm that runs in $\bigo(n^{k+1})$ time and space (where $n$ is the input size and $k$ the parameter) which implies that the problem is contained in XP~\cite{MWW_94}.

In this paper we will close this gap and show that \PP{} is complete for the complexity class XALP, which is a relatively new parameterized complexity class that was introduced by Bodlaender et al in~\cite{BGJ+_22}. We will show XALP-completeness by giving a reduction from the XALP-complete problem \TCMIS{}, using \TMG{} as a stepping stone. This makes \PP{} the first example of a ``natural'' problem that is XALP-complete and allows it to be used as a starting point for many other XALP-hardness proofs. Finally, we use the same reduction to give some lower bounds dependent on the Exponential Time Hypothesis.


\section{Definitions and Preliminary Results}
\label{sec:prelims}

All problems in this paper are parameterized. This means that the input contains a parameter separate from the rest of the input which allows us to analyze the runtime as a function of both the input and the parameter. If a parameterized problem with input size $n$ and parameter $k$ can be solved in $\bigo(f(k)n^c)$ time (with $f$ any computable function and $c$ any constant), we say that it is Fixed Parameter Tractable (FPT). A \emph{parameterized reduction} is an algorithm that transforms instances of one parameterized problem into instances of another parameterized problem, runs in FPT time, and whose new parameter is only dependent on the old parameter. A \emph{log-space reduction} is a parameterized reduction that additionally only uses $\bigo(f(k)\log(n))$ space. These reductions form the base of all parameterized complexity classes: all classes are defined up to equivalence under one of these reductions.

We use the following definition of \PP{}, which is a parameterized version of the original definition from Estabrook~\cite{EJM_76}.

\prblm{Perfect Phylogeny}{PP}{A set $G$ of genes, for each gene $g \in G$ a set $V_g$ of variants, and a set $S$ of species, where each species is defined as a tuple of gene-variants (exactly one per gene)}{The number of genes}{Does there exist a tree $T$ of species (not necessarily from $S$) that contains all species from $S$ and where the subtree of species containing a specific gene-variant is connected?}

\subsubsection{Triangulated and Colored Graphs.}

A graph is \emph{colored} if every vertex is assigned a color. The graph is \emph{properly colored} if there are no edges between vertices of the same color. For any cycle $C$ in a graph, a \emph{chord} is an edge between two vertices of $C$ that are not neighbors on $C$. A graph is \emph{triangulated} if every cycle of length at least four contains a chord. A \emph{triangulation} of a graph is a supergraph that is triangulated. We now define the problem \TCG, which was first given by Buneman~\cite{B_74}.

\prblm{Triangulating Colored Graphs}{TCG}{A colored graph $G$}{The number of colors used}{Does there exist a properly colored triangulation of $G$?}

We now introduce a multicolored variant of this problem. A graph is \emph{multicolored} if every vertex is assigned a (possibly empty) set of colors. The graph is \emph{properly multicolored} if there are no edges between vertices which share a color. This gives us the following problem:

\prblm{Triangulating Multicolored Graphs}{TMG}{A multicolored graph $G$}{The number of colors used}{Does there exist a properly multicolored triangulation of $G$?}

This problem is equivalent to \TCG{} under parameterized reductions. The general idea is to replace every multicolored vertex with a clique of normally colored vertices. A full proof is given in appendix \autoref{sec:multicolor}. We now define a tree decomposition and state some well-known properties of triangulated colored graphs.
\begin{definition}[Tree Decomposition]\label{def:treedecomp}
    Given a graph $G = (V,E)$, a tree decomposition is a tree $T$ where each vertex (bag) is associated with a subset of vertices from $T$. This tree must satisfy three conditions:
    \begin{itemize}
        \item For each vertex $v \in V$, there is at least one bag that contains $v$.
        \item For each edge $e \in V$, there is at least one bag that contains both endpoints of $e$.
        \item For each vertex $v \in V$, the subgraph of bags that contain $v$ is connected.
    \end{itemize}
\end{definition}

\begin{proposition}\label{prop:triang}
    Let $G$ be a (multi)colored graph and $C$ be a cycle.
    \begin{enumerate}[label=(\roman*)]
        \item\label{prop:triang:alternate} Suppose there exist two colors such that every vertex from $C$ is colored with at least one of these colors. Then $G$ admits no properly (multi)colored triangulation.
        \item\label{prop:triang:chord} Let $v$ be any vertex from $C$. In every triangulation of $G$ there is either an edge between $v$'s neighbors (in $C$) or a chord between $v$ and some non-neighbor vertex from $C$.
        \item\label{prop:triang:treedecomp} $G$ admits a properly colored triangulation if and only if $G$ admits a tree decomposition where each bag contains each color at most once.
    \end{enumerate}
\end{proposition}
\begin{proof}
    Omitted from main text. See appendix \autoref{sec:proofs}.
    \qed
\end{proof}

\subsubsection{XALP.}
A new complexity class in parameterized complexity theory is XALP~\cite{BGJ+_22}. Intuitively, it is the natural home of parameterized problems that are $W[t]$-hard for every $t$ and contain some hidden tree-structure. For \PP{}, this tree-structure is the required phylogeny. For \TCG{}, it is the tree decomposition arising from \autoref{prop:triang}\ref{prop:triang:treedecomp}.

Formally, XALP is the class of parameterized problems that are solvable on an alternating Turing machine using $\bigo(f(k) \log(n))$ memory and at most $\bigo(f(k) + \log(n))$ co-nondetermenistic computation steps, where $n$ is the input size and $k$ is the parameter. It is closed under log-space reductions. On Downey and Fellows' $W$-hierarchy, it lies between $W[t]$ and XP: XALP-hardness implies $W[t]$-hardness for every $t$ and XALP membership implies XP-membership.

An example of an XALP-complete problem is \TCMIS{}~\cite{BGJ+_22}. It is defined as a tree-chained variant of the well-known \MIS{} problem.

\prblm{Multicolor Independent Set}{MIS}{A colored graph $G$}{The number of colors used}{Does $G$ contain an independent set consisting of exactly one vertex of each color?}

\prblm{Tree-Chained Multicolor Independent Set}{TCMIS}{A binary tree $T$, for each vertex (bag) $B \in T$ a colored graph $G_B = (V_B,E_B)$ which we view as an instance of \MIS{}, and for each edge $e \in T$ a set of extra edges $E_e$ between the graphs corresponding to the endpoints of $e$.}{The maximum number of colors used in each instance of MIS}{Does there exist a solution to each instance of MIS such that for each of the extra edges at most one of the endpoints is contained in the solution?}


\section{Main Results}\label{sec:main}
In this section we will state the main result and explore some of its corollaries. We postpone the proof to the next sections.

\begin{theorem}\label{thm:xalpmember}
    \TCG{} is contained in XALP.
\end{theorem}
\begin{theorem}\label{thm:xalphard}
    There exists a log-space reduction from \TCMIS{} to \TMG{}. This reduction has a linear change of parameter.
\end{theorem}
We will prove \autoref{thm:xalpmember} in \autoref{sec:membership}. For \autoref{thm:xalphard}, we describe the reduction in \autoref{sec:hardness} and prove correctness of this reduction in appendix \autoref{sec:fullproof}.
\begin{theorem}[Main Result]
    The problems \PP{}, \TCG{} and \TMG{} are all XALP-complete.
\end{theorem}
\begin{proof}
    Combine \autoref{thm:xalpmember}, \autoref{thm:xalphard} and the equivalences between these three problems.
    \qed
\end{proof}

We now use these complexity results to show some lower bounds on the space and time usage of \PP{}. In the remainder of this section, let $n$ be the input size, $k$ the parameter, $f$ any computable function and $c$ any constant. We start with a bound on the runtime based on the Exponential Time Hypothesis.
\begin{proposition}
    Assuming ETH, the problems \PP{}, \TCG{} and \TMG{} cannot be solved in $f(k)n^{o(k)}$ time.
\end{proposition}
\begin{proof}
    We use as a starting point that, assuming ETH, the problem \MIS{} cannot be solved in $f(k)n^{o(k)}$ time~\cite{CCF+_04}. A trivial reduction to \TCMIS{} using a single-vertex tree then shows the same for that problem. Since the reduction given in \autoref{thm:xalphard} has a linear change in parameter we obtain the same lower bound for \TMG. Finally, using the equivalences proven in \autoref{sec:multicolor} and the known equivalences between TCG and PP (all with no change in parameter), the result follows.
    \qed
\end{proof}

We now bound the space usage based on the Slice-wise Polynomial Space Conjecture (SPSC). This conjectures that \textsc{Longest Common Subsequence} cannot be solved in both $n^{f(k)}$ time and $f(k)n^c$ space~\cite{PW_18}.
\begin{corollary}
    Assuming SPSC, the problems \PP{}, \TCG{} or \TMG{} cannot be solved in both $n^{f(k)}$ time and $f(k)n^c$ space.
\end{corollary}
\begin{proof}
    This proof uses the parameterized complexity class XNLP, which is defined as the class of parameterized problems that are solvable on a determenistic Turing machine using $\bigo(f(k)\log(n))$ memory. Comparing this with the definition of XALP shows that XALP-hardness implies XNLP-hardness. Since \textsc{Largest Common Subsequence} is XNLP-complete~\cite{EST_15}, SPSC applies to all XNLP-hard problems and consequently also to all XALP-hard problems such as the three problems from this corollary.
    \qed
\end{proof}

Compared with the existing algorithm that runs in $\bigo(n^{k+1})$ time and space~\cite{MWW_94}, these are close but not tight gaps.


\section{XALP Membership of \TCG{}}\label{sec:membership}
In this section we will prove \autoref{thm:xalpmember}.

Recall that \TCG{} asks us to determine whether a colored graph can be triangulated. Because of \autoref{prop:triang}\ref{prop:triang:treedecomp}, this is equivalent to finding a tree decomposition where each bag contains each color at most once. We claim that it is equivalent to find a tree decomposition where each bag contains each color \emph{exactly} once.
\begin{lemma}\label{lemma:exact}
    A colored graph admits a tree decomposition where each bag contains each color at most once, if and only if it admits a tree decomposition where each bag contains each color exactly once.
\end{lemma}
\begin{proof}
    Omitted. See appendix \autoref{sec:proofs}.
    \qed
\end{proof}

We can now prove XALP membership.
\begin{proof}[of \autoref{thm:xalpmember}]
    We construct an alternating Turing machine (ATM) that, given an instance of \TCG{}, determines whether there exists a tree decomposition that contains each color exactly once. As a refresher, an ATM is a Turing machine that has access to both nondetermenistic and co-nondeterministic branching steps. A nondetermenistic step leads to ACCEPT if at least one successor state leads to ACCEPT and a co-nondetermenistic step leads to ACCEPT if all successor states lead to ACCEPT. 

    Our Turing machine is based on the XP-time algorithm we mentioned before~\cite{MWW_94}. We use the following claim without proof: given a graph $G$, a deterministic Turing machine can determine the whether two vertices belong to the same connected component in logarithmic space and polynomial time~\cite{R_08}. Repeated application of this result allows us to branch on all connected components of a graph using several co-nondeterministic steps.

    Let $G$ be any colored graph. The Turing machine will use nondeterministic steps to determine how to modify each bag compared to its parent and co-nondeterministic steps to simultaneously verify all subtrees. A precise formulation is given below:
    \begin{itemize}
        \item Using $k$ nondeterministic steps, determine an initial bag $S$ which contains one vertex of each color. During computations that lead to ACCEPT, each $S$ will be a bag from the tree decomposition.
        \item Keep track of some vertex $i$ that is initially NULL. This will signify the parent of the current bag $S$.
        \item Repeat the following until an ACCEPT or REJECT state is reached:
            \begin{itemize}
                \item Determine all components of $G \setminus S$. Using a co-nondeterministic step, we branch into every component except the one that contains $i$. If this results in zero branches (e.g. when there are no other components), ACCEPT.
                \item Let $C$ be the component our current branch is in. We determine a vertex $v \in C$ with a nondeterministic step.
                \item Determine the vertex $w \in S$ that has the same color as $v$. Since $S$ contains one vertex of every color, $w$ exists.
                \item If $w$ is adjacent to any vertex from $C$, REJECT. This means that the current guess for how to modify $S$ is incorrect.
                \item Modify $S$ by adding $v$ and removing $w$. Set $i$ to $w$.
            \end{itemize}
    \end{itemize}
    Overall, this alternating Turing machine constructively determines a rooted tree decomposition if one exists and thus solves \TCG{}. It also satisfies the memory requirement: the only memory usage is the set $S$, a constant number of extra vertices, and the memory needed to branch on connected components. Since memory of a vertex uses $\bigo(\log(n))$ space and $\abs{S} = k$, we need $\bigo(k\log(n))$ space. We also use polynomial time: the time usage in the computation of each bag is a constant plus the time needed to find the connected components which results in polynomial time overall. Finally, we require at most $\bigo(n)$ co-nondeterministic computation steps: each co-nondeterministic step corresponds to branching into a subtree of the eventual (rooted) tree decomposition. Since each subtree introduces at least one vertex that is used nowhere else in the tree, there are at most $\bigo(n)$ subtrees.

    Overall, we conclude that \TCG{} is contained in XALP.
    \qed
\end{proof}


\section{Zipper Chains and Gadgets}\label{sec:zipper}
In this section we will introduce two multicolored graph components, the \emph{zipper chain} and the \emph{zipper gadget}. Their most important property is \autoref{prop:zipper-gadget} which says that a zipper gadget has a fixed number of triangulations. This will be used in the XALP-hardness proof to represent a choice.
\begin{definition}
    A zipper chain is a multicolored graph that consists of two paths $P$ and $Q$, not necessarily of the same length. The vertices of $P$ and $Q$ are respectively labeled as $p_1, p_2, \ldots$ and $q_1, q_2, \ldots$.

    The vertices are colored in 7 colors, with 2 colors per vertex. For ease of explanation, the colors are grouped in three groups with sizes 1, 2, and 4. The first group contains one color $a$ which is added to odd-labeled vertices from $P$ and even-labeled vertices from $Q$. The second group contains the color $b_P$ which is added to even-labeled vertices of $P$ and the color $b_Q$ which is added to odd-labeled vertices of $Q$. The third group contains four colors $c_1, c_2, c_3$ and $c_4$ where $c_i$ is added to vertices in $P$ whose index is equivalent to $i$ (mod 4) and vertices in $Q$ whose index is equivalent to $i+2$ (mod 4).
\end{definition}
\begin{figure}[b]
    \centering
    \begin{tikzpicture}
        \begin{scope}[every node/.style={circle,thick,draw,inner sep=0pt,minimum size=0.6cm}]
            \node (p1) at (0,1) {\tiny$ac_1$};
            \node (p2) at (1,1) {\tiny$b_Pc_2$};
            \node (p3) at (2,1) {\tiny$ac_3$};
            \node (p4) at (3,1) {\tiny$b_Pc_4$};
            \node (p5) at (4,1) {\tiny$ac_1$};
            \node (p6) at (5,1) {\tiny$b_Pc_2$};
            \node (p7) at (6,1) {\tiny$ac_3$};
            \node (p8) at (7,1) {\tiny$b_Pc_4$};

            \node (q1) at (0,0) {\tiny$b_Qc_3$};
            \node (q2) at (1,0) {\tiny$ac_4$};
            \node (q3) at (2,0) {\tiny$b_Qc_1$};
            \node (q4) at (3,0) {\tiny$ac_2$};
            \node (q5) at (4,0) {\tiny$b_Qc_3$};
            \node (q6) at (5,0) {\tiny$ac_4$};
            \node (q7) at (6,0) {\tiny$b_Qc_1$};
            \node (q8) at (7,0) {\tiny$ac_2$};
        \end{scope}

        \begin{scope}[>={Stealth[black]},
                      every node/.style={fill=white,circle},
                      every edge/.style={draw=red,thick}]
            \draw[dotted] (-1,1) -- (p1);
            \draw (p1) -- (p2);
            \draw (p2) -- (p3);
            \draw (p3) -- (p4);
            \draw (p4) -- (p5);
            \draw (p5) -- (p6);
            \draw (p6) -- (p7);
            \draw (p7) -- (p8);
            \draw[dotted] (p8) -- (8,1);

            \draw[dotted] (-1,0) -- (q1);
            \draw (q1) -- (q2);
            \draw (q2) -- (q3);
            \draw (q3) -- (q4);
            \draw (q4) -- (q5);
            \draw (q5) -- (q6);
            \draw (q6) -- (q7);
            \draw (q7) -- (q8);
            \draw[dotted] (q8) -- (8,0);
        \end{scope}

        \node at (-1.5,1) {$P$:};
        \node at (-1.5,0) {$Q$:};

        \draw[decorate,decoration={brace}] (-0.2,1.5) -- (3.2,1.5) node[pos=0.5,above=5pt]{a tooth from $P$};
        \draw[decorate,decoration={brace}] (3.8,1.5) -- (7.2,1.5) node[pos=0.5,above=5pt]{another tooth from $P$};
        \draw[decorate,decoration={brace,mirror}] (-0.2,-0.5) -- (3.2,-0.5) node[pos=0.5,above=-20pt]{a tooth from $Q$};
        \draw[decorate,decoration={brace,mirror}] (3.8,-0.5) -- (7.2,-0.5) node[pos=0.5,above=-20pt]{another tooth from $Q$};
    \end{tikzpicture}
    \caption{A zipper chain.}\label{fig:zipper_chain}
\end{figure}
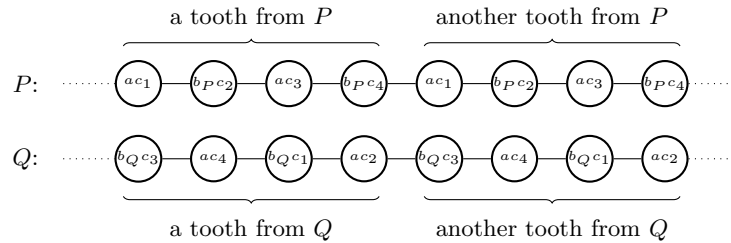

To summarize: the colors on path $P$ are $ac_1, b_Pc_2, ac_3, b_Pc_4, ac_1, \ldots$ and those of $Q$ are $b_Qc_3, ac_4, b_Qc_1, ac_2, b_Qc_3, \ldots$. This is visualized in \autoref{fig:zipper_chain}.

This color pattern repeats every four vertices. We call such a repetition a \emph{tooth} of the zipper chain. If a triangulation of the zipper chain contains an edge between some tooth of $P$ and some tooth of $Q$ and at least one endpoint of this edge contains the color $a$, we say that these two teeth are \emph{locked together}.

\begin{proposition}\label{prop:zipper}
    Let $G$ be a graph containing a zipper chain $(P,Q)$ and assume that there is a cycle that fully contains both $P$ and $Q$. Any triangulation of $G$ satisfies the following properties. Because of symmetry, all properties also hold with $P$ and $Q$ reversed.
    \begin{enumerate}[label=(\roman*)]
        \item\label{prop:zipper:vertical} There is no edge between two non-adjacent vertices of $P$.
        \item\label{prop:zipper:convex} If there exist two edges between $P$ and $Q$ which share an endpoint in $P$, then the common endpoint in $P$ is connected to all vertices of $Q$ that lie between the other two endpoints.
        \item\label{prop:zipper:climbing} If $(p_i,q_j)$ is an edge, then either $(p_{i+1},q_j)$ or $(p_i,q_{j+1})$ is also an edge (as long as either $p_{i+1}$ or $q_{j+1}$ exists).
        \item\label{prop:zipper:sliding} If $(p_i,q_j)$ is an edge and $p_i$ contains the color $a$, then $(p_{i+1},q_{j+1})$ is also an edge (as long as both $p_{i+1}$ and $q_{j+1}$ exist). Here, $q_{j+1}$ contains the color $a$.
        \item\label{prop:zipper:teeth} If the $i$-th tooth of $P$ and the $j$-th tooth of $Q$ are locked together, then the $i+1$-th tooth of $P$ and the $j+1$-th tooth of $Q$ are also locked together.
        \item\label{prop:zipper:unique} Each tooth from $P$ is locked together with at most one tooth from $Q$.
    \end{enumerate}
\end{proposition}
\begin{figure}[t]
    \centering
    \begin{tikzpicture}
        \begin{scope}[every node/.style={circle,thick,draw,inner sep=0pt,minimum size=0.6cm}]
            \node (p1) at (0,1) {\tiny$ac_1$};
            \node (p2) at (1,1) {\tiny$b_Pc_2$};
            \node (p3) at (2,1) {\tiny$ac_3$};
            \node (p4) at (3,1) {\tiny$b_Pc_4$};
            \node (p5) at (4,1) {\tiny$ac_1$};
            \node (p6) at (5,1) {\tiny$b_Pc_2$};
            \node (p7) at (6,1) {\tiny$ac_3$};
            \node (p8) at (7,1) {\tiny$b_Pc_4$};

            \node (q1) at (0,0) {\tiny$b_Qc_3$};
            \node (q2) at (1,0) {\tiny$ac_4$};
            \node (q3) at (2,0) {\tiny$b_Qc_1$};
            \node (q4) at (3,0) {\tiny$ac_2$};
            \node (q5) at (4,0) {\tiny$b_Qc_3$};
            \node (q6) at (5,0) {\tiny$ac_4$};
            \node (q7) at (6,0) {\tiny$b_Qc_1$};
            \node (q8) at (7,0) {\tiny$ac_2$};
        \end{scope}

        \begin{scope}[>={Stealth[black]},
                      every node/.style={fill=white,circle},
                      every edge/.style={draw=red,thick}]
            \draw[dotted] (-1,1) -- (p1);
            \draw (p1) -- (p2);
            \draw (p2) -- (p3);
            \draw (p3) -- (p4);
            \draw (p4) -- (p5);
            \draw (p5) -- (p6);
            \draw (p6) -- (p7);
            \draw (p7) -- (p8);
            \draw[dotted] (p8) -- (8,1);

            \draw[dotted] (-1,0) -- (q1);
            \draw (q1) -- (q2);
            \draw (q2) -- (q3);
            \draw (q3) -- (q4);
            \draw (q4) -- (q5);
            \draw (q5) -- (q6);
            \draw (q6) -- (q7);
            \draw (q7) -- (q8);
            \draw[dotted] (q8) -- (8,0);

            \draw (-0.5,0.5) -- (q1);
            \draw (p1) -- (q1);
            \draw (p2) -- (q1);
            \draw (p2) -- (q2);
            \draw (p2) -- (q3);
            \draw (p3) -- (q3);
            \draw (p4) -- (q3);
            \draw (p4) -- (q4);
            \draw (p4) -- (q5);
            \draw (p5) -- (q5);
            \draw (p6) -- (q5);
            \draw (p6) -- (q6);
            \draw (p6) -- (q7);
            \draw (p7) -- (q7);
            \draw (p8) -- (q7);
            \draw (p8) -- (q8);
            \draw (p8) -- (7.5,0.5);
        \end{scope}

        \node at (-1.5,1) {$P$:};
        \node at (-1.5,0) {$Q$:};
    \end{tikzpicture}
    \caption{A possible triangulation of a zipper chain.}\label{fig:triangulated_zipper_chain}
\end{figure}
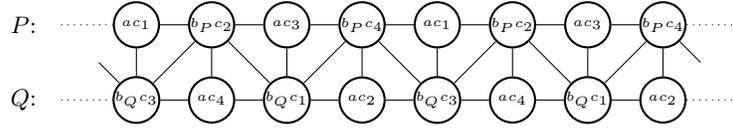
\begin{proof}
    We prove the statements in order.
    \begin{enumerate}[label=(\roman*)]
        \item If, to the contrary, such an edge does exist, then this edge together with the rest of $P$ forms a cycle whose vertices alternate between the colors $a$ and $b_P$. Because of \autoref{prop:triang}\ref{prop:triang:alternate} such a cycle cannot be triangulated.
        \item Because of part \ref{prop:zipper:vertical}, the cycle formed by these two edges and the path between the two endpoints on $Q$ can only be triangulated by adding edges with an endpoint in $P$.
        \item If $(p_{i+1},q_j)$ is not an edge then \autoref{prop:triang}\ref{prop:triang:chord} shows that $p_i$ must be connected to another vertex in the cycle. Because of part \ref{prop:zipper:vertical} this neighbor is a vertex from $Q$. Because of part \ref{prop:zipper:convex} $p_i$ must then also be connected to $q_{j+1}$.
        \item Without loss of generality, say that $p_i$ also contains the color $c_1$. Then, $q_j$ must have the colors $b_Q$ and $c_3$: all other color combinations share a color with $p_i$. Since $q_{j+1}$ and $p_i$ both contain the color $a$ there is no edge between them so part \ref{prop:zipper:climbing} implies that there is one between $p_{i+1}$ and $q_j$. Since $p_{i+2}$ and $q_j$ share the color $c_3$, the same argument implies that there is an edge between $p_{i+1}$ and $q_{j+1}$.
        \item Apply part \ref{prop:zipper:sliding} four times to the edge connecting the $i$-th and $j$-th teeth of $P$ and $Q$ (respectively) to obtain an edge connecting the $i+1$-th and $j+1$-th teeth of $P$ and $Q$ (respectively).
        \item Suppose to the contrary that a tooth from $P$ is locked together with two teeth from $Q$. After some applications of parts \ref{prop:zipper:climbing} and \ref{prop:zipper:sliding} we find that the last vertex from the tooth from $P$ (which has colors $b_P$ and $c_3$) is connected to the last vertex from both teeth from $Q$. Because of part \ref{prop:zipper:convex}, it is connected to all four vertices of the tooth from $Q$ with the higher index. At least one of these also contains the color $c_3$ so this is a contradiction.
    \end{enumerate}
    \qed
\end{proof}

\subsubsection{Zipper Gadgets.}

We now introduce the \emph{zipper gadget}. It is a zipper chain with a specific length and a head and tail. An example is given in \autoref{fig:zipper_gadget}.
\begin{definition}
    A zipper gadget of size $n$ and skew $s$ (satisfying $n>0,s \geq 0$) is a zipper chain with the following modifications:
    \begin{itemize}
        \item The path $P$ contains $4n-1$ vertices, and thus $n$ teeth. The last tooth misses one vertex.
        \item The path $Q$ contains $4(n+s)$ vertices, and thus $n+s$ teeth.
        \item There are two additional vertices with just the color $b_P$: a head $h$ and a tail $t$. The head is connected to the first vertices of $P$ and $Q$ and the tail to the last vertices of $P$ and $Q$.
    \end{itemize}
\end{definition}
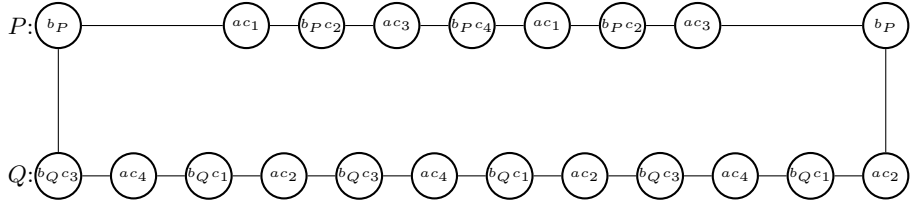
\begin{figure}[t]
    \centering
    \begin{tikzpicture}
        \begin{scope}[every node/.style={circle,thick,draw,inner sep=0pt,minimum size=0.6cm}]
            \node (h) at (0,2) {\tiny$b_P$};
            \node (t) at (11,2) {\tiny$b_P$};

            \node (p1) at (2.5,2) {\tiny$ac_1$};
            \node (p2) at (3.5,2) {\tiny$b_Pc_2$};
            \node (p3) at (4.5,2) {\tiny$ac_3$};
            \node (p4) at (5.5,2) {\tiny$b_Pc_4$};
            \node (p5) at (6.5,2) {\tiny$ac_1$};
            \node (p6) at (7.5,2) {\tiny$b_Pc_2$};
            \node (p7) at (8.5,2) {\tiny$ac_3$};

            \node (q1) at (0,0) {\tiny$b_Qc_3$};
            \node (q2) at (1,0) {\tiny$ac_4$};
            \node (q3) at (2,0) {\tiny$b_Qc_1$};
            \node (q4) at (3,0) {\tiny$ac_2$};
            \node (q5) at (4,0) {\tiny$b_Qc_3$};
            \node (q6) at (5,0) {\tiny$ac_4$};
            \node (q7) at (6,0) {\tiny$b_Qc_1$};
            \node (q8) at (7,0) {\tiny$ac_2$};
            \node (q9) at (8,0) {\tiny$b_Qc_3$};
            \node (q10) at (9,0) {\tiny$ac_4$};
            \node (q11) at (10,0) {\tiny$b_Qc_1$};
            \node (q12) at (11,0) {\tiny$ac_2$};
        \end{scope}

        \begin{scope}[>={Stealth[black]},
                      every node/.style={fill=white,circle},
                      every edge/.style={draw=red,thick}]
            \draw (h) -- (p1);
            \draw (p1) -- (p2);
            \draw (p2) -- (p3);
            \draw (p3) -- (p4);
            \draw (p4) -- (p5);
            \draw (p5) -- (p6);
            \draw (p6) -- (p7);
            \draw (p7) -- (t);

            \draw (h) -- (q1);
            \draw (q1) -- (q2);
            \draw (q2) -- (q3);
            \draw (q3) -- (q4);
            \draw (q4) -- (q5);
            \draw (q5) -- (q6);
            \draw (q6) -- (q7);
            \draw (q7) -- (q8);
            \draw (q8) -- (q9);
            \draw (q9) -- (q10);
            \draw (q10) -- (q11);
            \draw (q11) -- (q12);
            \draw (q12) -- (t);
        \end{scope}

        \node at (-0.5,2) {$P$:};
        \node at (-0.5,0) {$Q$:};
    \end{tikzpicture}
    \caption{A zipper gadget of size 2 and skew 1.}\label{fig:zipper_gadget}
\end{figure}
\begin{proposition}\label{prop:zipper-gadget}
    There are exactly $s+1$ ways to triangulate a zipper gadget with skew $s$. These ways are identified by the offset at which the teeth lock together.
\end{proposition}
\begin{proof}
    Observe that the entire gadget forms a cycle, so \autoref{prop:zipper} applies. Consider a vertex from $P$ that contains the color $a$. Its neighbors share the color $b_P$, so \autoref{prop:triang}\ref{prop:triang:chord} shows that this vertex must be connected to some other vertex from the cycle. This cannot be $h$, $t$ or another vertex from $P$ since that would introduce a cycle containing only the colors $a$ and $b_P$. Hence, the other endpoint must be a vertex from $Q$. This shows that each tooth from $P$ is locked together with at least one tooth from $Q$.

    \autoref{prop:zipper}\ref{prop:zipper:unique} now shows that each tooth from $P$ is locked together with exactly one tooth from $Q$. Let $\Delta$ be the index of the tooth locked together with the first tooth of $P$. \autoref{prop:zipper}\ref{prop:zipper:teeth} now shows that any tooth with index $i$ must be connected to tooth $i+\Delta$. Since $Q$ has $s$ more teeth than $P$, the offset $\Delta$ must be between 0 and $s$. We conclude that there are at most $s+1$ ways to triangulate a zipper gadget with offset $s$ and that these ways are identified by the offset.

    To complete the proof, we now show that each case can actually be extended into a triangulation of the zipper gadget. Let $\Delta$ be the target offset. We add the following edges:
    \begin{itemize}
        \item An edge between the head $h$ and every vertex from the first $\Delta$ teeth from $Q$.
        \item Edges between the $i$-th tooth from path $P$ and the $i+\Delta$-th tooth from $Q$ according to the pattern described in parts \ref{prop:zipper:climbing} and \ref{prop:zipper:sliding} of \autoref{prop:zipper}. This includes one overlap edge between the last vertex of each tooth from $P$ and the first vertex from the next tooth from $Q$.
        \item An edge between the tail $t$ and every vertex from the last $s-\Delta$ teeth from $Q$.
    \end{itemize}
    An example of such a triangulation is given in \autoref{fig:triangulated_zipper_gadget}. One can observe that this construction indeed triangulates the zipper gadget.
    \qed
\end{proof}
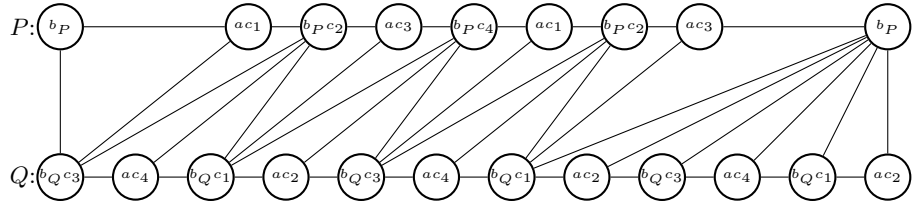
\begin{figure}[t]
    \centering
    \begin{tikzpicture}
        \begin{scope}[every node/.style={circle,thick,draw,inner sep=0pt,minimum size=0.6cm}]
            \node (h) at (0,2) {\tiny$b_P$};
            \node (t) at (11,2) {\tiny$b_P$};

            \node (p1) at (2.5,2) {\tiny$ac_1$};
            \node (p2) at (3.5,2) {\tiny$b_Pc_2$};
            \node (p3) at (4.5,2) {\tiny$ac_3$};
            \node (p4) at (5.5,2) {\tiny$b_Pc_4$};
            \node (p5) at (6.5,2) {\tiny$ac_1$};
            \node (p6) at (7.5,2) {\tiny$b_Pc_2$};
            \node (p7) at (8.5,2) {\tiny$ac_3$};

            \node (q1) at (0,0) {\tiny$b_Qc_3$};
            \node (q2) at (1,0) {\tiny$ac_4$};
            \node (q3) at (2,0) {\tiny$b_Qc_1$};
            \node (q4) at (3,0) {\tiny$ac_2$};
            \node (q5) at (4,0) {\tiny$b_Qc_3$};
            \node (q6) at (5,0) {\tiny$ac_4$};
            \node (q7) at (6,0) {\tiny$b_Qc_1$};
            \node (q8) at (7,0) {\tiny$ac_2$};
            \node (q9) at (8,0) {\tiny$b_Qc_3$};
            \node (q10) at (9,0) {\tiny$ac_4$};
            \node (q11) at (10,0) {\tiny$b_Qc_1$};
            \node (q12) at (11,0) {\tiny$ac_2$};
        \end{scope}

        \begin{scope}[>={Stealth[black]},
                      every node/.style={fill=white,circle},
                      every edge/.style={draw=red,thick}]
            \draw (h) -- (p1);
            \draw (p1) -- (p2);
            \draw (p2) -- (p3);
            \draw (p3) -- (p4);
            \draw (p4) -- (p5);
            \draw (p5) -- (p6);
            \draw (p6) -- (p7);
            \draw (p7) -- (t);

            \draw (h) -- (q1);
            \draw (q1) -- (q2);
            \draw (q2) -- (q3);
            \draw (q3) -- (q4);
            \draw (q4) -- (q5);
            \draw (q5) -- (q6);
            \draw (q6) -- (q7);
            \draw (q7) -- (q8);
            \draw (q8) -- (q9);
            \draw (q9) -- (q10);
            \draw (q10) -- (q11);
            \draw (q11) -- (q12);
            \draw (q12) -- (t);

            \draw (p1) -- (q1);
            \draw (p2) -- (q1);
            \draw (p2) -- (q2);
            \draw (p2) -- (q3);
            \draw (p3) -- (q3);
            \draw (p4) -- (q3);
            \draw (p4) -- (q4);
            \draw (p4) -- (q5);
            \draw (p5) -- (q5);
            \draw (p6) -- (q5);
            \draw (p6) -- (q6);
            \draw (p6) -- (q7);
            \draw (p7) -- (q7);
            \draw (t) -- (q7);
            \draw (t) -- (q8);
            \draw (t) -- (q9);
            \draw (t) -- (q10);
            \draw (t) -- (q11);
        \end{scope}

        \node at (-0.5,2) {$P$:};
        \node at (-0.5,0) {$Q$:};
    \end{tikzpicture}
    \caption{One of the two triangulations of the zipper gadget from \autoref{fig:zipper_gadget}. This one has offset 0.}\label{fig:triangulated_zipper_gadget}
\end{figure}


\section{XALP-hardness of \TMG{}}\label{sec:hardness}
In this section we describe the reduction from \autoref{thm:xalphard}. The intuition is as follows. We want to reduce from \TCMIS{}, which comes down to selecting a vertex from each color for each instance of \MIS{}. These choices must be compatible: we may not choose two vertices which share an edge. The selection of a vertex will be done by creating zipper gadgets and interpreting each possible triangulation as a choice of a vertex. The compatibility checks will be done by combining two zipper gadgets in a way that makes it impossible to simultaneously triangulate both zipper gadgets in the respective choices. This construction borrows a technique, namely on how to create and combine gadgets from the TCMIS tree, from the XALP-completeness proof for \textsc{Tree Partition Width} from~\cite{BGJ_22}. The actual gadgets and their combination procedure are new.

Let an instance of TCMIS be given. Let $T$ and $k$ be the (binary) tree and parameter from this instance. For any node $n \in T$, we have an associated instance of \MIS{} consisting of a set of vertices $S_c$ for each color $c$. Without loss of generality, we can assume that all sets $S_c$ have the same size, say $r+1$: if not, then we can add extra vertices to $S_c$ that are connected to all other vertices and thus never occur in an independent set. We also assume that $S_c$ is ordered in some way. This allows us to refer to vertices as $v_{n,c,i}$ where $n$ is the node from $T$, $c$ is the color, and $i$ is the index in $S_c$ (which, for ease of explanation, is zero-based). We also have a set of edges $E$, which we again assume to be ordered in some way. Each edge connects two vertices $v_{n_1,c_1,i_1}$ and $v_{n_2,c_2,i_2}$ where $n_1$ and $n_2$ are either the same node or neighbors in $T$ and where $c_1$ and $c_2$ are distinct if $n_1=n_2$. Let $m := \abs{E}$ be the total number of edges.

First, we transform $T$ into a rooted tree $T'$ by choosing any node $u \in T$, adding two new nodes $v$ and $w$ and two edges $(u,v)$ and $(v,w)$, and setting $w$ as the root. This way, each node from the original tree $T$ has a parent and a grandparent in $T'$. We now construct a graph $G$ which will be an instance of TMG. It will consist of several zipper gadgets in which some vertices have been identified with each other: that is, where some vertices with distinct colors are merged into one vertex with the combined set of colors. For each node $n$ in $T$ and each color $c$ in its associated instance of \MIS{}, we add a zipper gadget $z_{n,c}$ of size $2mr+1$ and skew $r$. The middle tooth of path $P$ (with index $mr+1$) is special: we call it the \emph{middle}. We now say that this zipper gadget starts in $n$, passes through the parent of $n$ and ends in the grandparent of $n$. This is supported with some vertex identifications: for each node $n$ in $T'$, we identify the heads of all zipper gadgets starting at $n$, the tails of all zipper gadgets ending at $n$, and the last vertex of the middles (with colors $c_4$ and $b_P$) of all zipper gadgets that pass through $n$. Observe that the path $P$ of each zipper gadget now consists of $m$ sets of $r$ teeth between its head and middle, and also $m$ sets of $r$ teeth between its middle and tail.

Each zipper gadget is assigned its own set of 7 colors such that no two zipper gadgets which start, pass through, or end in a common node share a color. We claim that this can be done using at most $7k$ sets of 7 colors. Assign colors to nodes in order of distance to the root of $T'$ (closest to the root first). Let $n$ be the current node. All zipper gadgets that have already been assigned colors and intersect with zipper gadgets starting from $n$ are those that start at either: $n$'s parent, the other child of $n$'s parent ($n$'s sibling), $n$'s grandparent, the other child of $n$'s grandparent ($n$'s uncle), or any of the two children from that vertex ($n$'s cousins). In total, this is at most $6k$ other zipper gadgets. To color the $k$ zipper gadgets starting at $n$, we can thus use the remaining $7k-6k = k$ sets of 7 colors.

In a triangulation of $G$, each zipper gadget will represent a choice of a vertex from $S_c$: if the zipper gadget is triangulated with offset $\Delta$, then we choose the vertex with index $\Delta$ from $S_c$. Each of the $m$ sets of $r$ teeth between head and middle or between middle and tail will represent a restriction regarding one of the edges. For each edge $e_i$ (with index $i$) with endpoints $v_{n_1,c_1,i_1}$ and $v_{n_2,c_2,i_2}$ we want to exclude the possibility of simultaneously triangulating the zipper gadget $z_{n_1,c_1}$ with offset $i_1$ and the zipper gadget $z_{n_2,c_2}$ with offset $i_2$. This is done as follows.

Let $(P_1,Q_1)$ and $(P_2,Q_2)$ be the paths which form the zipper gadgets. We now identify two vertices from $P_1$ and $P_2$ and add a new color $d$ to some vertices from $Q_1$ and $Q_2$. This is visualized in \autoref{fig:zipper_merge}. The idea is that if we would triangulate both zipper gadgets in a way that adds edges between the vertices with color $d$ and the merged vertex, then any triangulation of both zipper gadgets together forces an edge between the vertices with the color $d$ which is impossible. We now describe exactly which vertices should be modified.

We consider two cases: either $n_1$ and $n_2$ are the same node or they are neighbors in $T$. In the first case, we consider the tooth with index $ir$ from both $P_1$ and $P_2$ and identify the first vertex from these teeth with each other. We also consider tooth $ir+i_1$ from $Q_1$ and tooth $ir+i_2$ from $Q_2$ and add a new color $d$ to the first vertex of these teeth. In the second case, we assume without loss of generality that $n_1$ is the parent of $n_2$. We do almost the same as in the first case, except that we use the second half of the zipper gadget $z_{n_2,c_2}$: we identify the first vertex of tooth $ir$ from $P_1$ and tooth $mr+1+ir$ from $P_2$, and we add color $d$ to the first vertex of tooth $ir+i_1$ from $Q_1$ and tooth $mr+1+ir+i_2$ from $Q_2$.

This completes the construction. Observe that this construction uses $49k+1$ colors ($7k$ sets of 7 colors for the zipper gadgets and one for the extra color $d$) and thus that the change in parameter is linear. Also observe that the construction can be performed in logarithmic working space since the creation and merging of the zipper gadgets only require local information from the original TCMIS instance. This shows that we indeed have a logspace reduction.

The proof that this TMG instance admits a triangulation if and only if the original TCMIS instance admits a solution is a direct result of the intuitive insights mentioned during the construction and thus omitted from the main text. A full proof is given in appendix \autoref{sec:fullproof}.
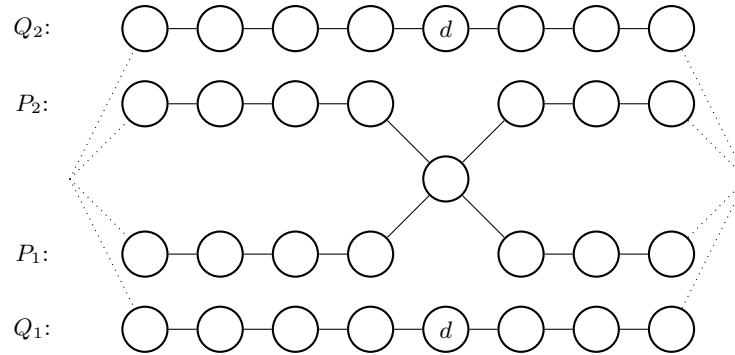
\begin{figure}[t]
\centering
\begin{tikzpicture}
    \begin{scope}[every node/.style={circle,thick,draw,inner sep=0pt,minimum size=0.6cm}]
        \node (q11) at (0,0) {};
        \node (q12) at (1,0) {};
        \node (q13) at (2,0) {};
        \node (q14) at (3,0) {};
        \node (q15) at (4,0) {\small$d$};
        \node (q16) at (5,0) {};
        \node (q17) at (6,0) {};
        \node (q18) at (7,0) {};

        \node (p11) at (0,1) {};
        \node (p12) at (1,1) {};
        \node (p13) at (2,1) {};
        \node (p14) at (3,1) {};
        \node (p16) at (5,1) {};
        \node (p17) at (6,1) {};
        \node (p18) at (7,1) {};

        \node (p5) at (4,2) {};

        \node (p21) at (0,3) {};
        \node (p22) at (1,3) {};
        \node (p23) at (2,3) {};
        \node (p24) at (3,3) {};
        \node (p26) at (5,3) {};
        \node (p27) at (6,3) {};
        \node (p28) at (7,3) {};

        \node (q21) at (0,4) {};
        \node (q22) at (1,4) {};
        \node (q23) at (2,4) {};
        \node (q24) at (3,4) {};
        \node (q25) at (4,4) {\small$d$};
        \node (q26) at (5,4) {};
        \node (q27) at (6,4) {};
        \node (q28) at (7,4) {};
    \end{scope}

    \begin{scope}[>={Stealth[black]},
                  every node/.style={fill=white,circle},
                  every edge/.style={draw=red,thick}]
        \draw[dotted] (-1,2) -- (p11);
        \draw (p11) -- (p12);
        \draw (p12) -- (p13);
        \draw (p13) -- (p14);
        \draw (p14) -- (p5);
        \draw (p5) -- (p16);
        \draw (p16) -- (p17);
        \draw (p17) -- (p18);
        \draw[dotted] (p18) -- (8,2);

        \draw[dotted] (-1,2) -- (q11);
        \draw (q11) -- (q12);
        \draw (q12) -- (q13);
        \draw (q13) -- (q14);
        \draw (q14) -- (q15);
        \draw (q15) -- (q16);
        \draw (q16) -- (q17);
        \draw (q17) -- (q18);
        \draw[dotted] (q18) -- (8,2);

        \draw[dotted] (-1,2) -- (p21);
        \draw (p21) -- (p22);
        \draw (p22) -- (p23);
        \draw (p23) -- (p24);
        \draw (p24) -- (p5);
        \draw (p5) -- (p26);
        \draw (p26) -- (p27);
        \draw (p27) -- (p28);
        \draw[dotted] (p28) -- (8,2);

        \draw[dotted] (-1,2) -- (q21);
        \draw (q21) -- (q22);
        \draw (q22) -- (q23);
        \draw (q23) -- (q24);
        \draw (q24) -- (q25);
        \draw (q25) -- (q26);
        \draw (q26) -- (q27);
        \draw (q27) -- (q28);
        \draw[dotted] (q28) -- (8,2);
    \end{scope}

    \node at (-1.5,1) {$P_1$:};
    \node at (-1.5,0) {$Q_1$:};
    \node at (-1.5,3) {$P_2$:};
    \node at (-1.5,4) {$Q_2$:};
\end{tikzpicture}
\caption{How two zipper gadgets are combined: two vertices from $P_1$ and $P_2$ are merged into one vertex and two vertices from $Q_1$ and $Q_2$ are given the extra color $d$.}\label{fig:zipper_merge}
\end{figure}


\section{Future Research}\label{sec:future}
Let $n$ be the input size, $k$ the parameter, $f$ any computable function, $c$ any constant, and $\epsilon$ any small positive constant. We have shown that \PP{} and \TCG{} are XALP-complete and that (assuming ETH) there exist no algorithms that solve any of them in $f(k)n^{o(k)}$ time. This increases the number of ``natural'' problems in the complexity class XALP and gives more reason to determine properties of this complexity class. Additionally, these problems can be used as a starting point for XALP-hardness reductions for other parameterized problems.

Another future research direction might be to close or reduce the gaps between the current upper and lower bounds on space and time usage. For the time gap, there is a lower bound of $f(k)n^{o(k)}$ (assuming ETH) and an upper bound of $\bigo(n^{k+1})$~\cite{MWW_94}. For the space gap on algorithms that run in $n^{f(k)}$ time, there is a lower bound of $f(k)n^c$ (assuming SPSC) and an upper bound of again $\bigo(n^{k+1})$~\cite{MWW_94}. One way to close the time gap could be by assuming the Strong Exponential Time Hypothesis (SETH). We expect that, assuming SETH, a lower bound like $f(k)n^{k-\epsilon}$ should be possible.

We also rule out a research direction. Triangulating a colored graph comes down to finding a tree decomposition where each bag contains each color at most once. A similar problem would be to instead look for a \emph{path} decomposition where each bag contains each color at most one. This problem, known as \textsc{Intervalizing Colored Graphs}, is already NP-complete for the case $k=4$~\cite{BF_96}.


%
%

\bibliographystyle{splncs04}
\bibliography{perfect-phylogeny}

\appendix


\section{Triangulating Multicolored Graphs}\label{sec:multicolor}

In this appendix section we will prove the equivalence of \TCG{} (TCG) and \TMG{} (TMG). We begin with two lemmas.
\begin{lemma}\label{lemma:full}
    Let $G$ be a graph and $T$ a tree decomposition of $G$. Then for each clique $C$ there is a bag in $T$ which fully contains $C$.
\end{lemma}
\begin{proof}
    We use induction on the size of $C$. For $\abs{C} = 1$ and $\abs{C} = 2$, the result follows directly from the definition of tree decomposition. Now suppose that $\abs{C} > 2$ and let $v$ be a vertex from $C$. By induction hypothesis, there must be a bag $B_r$ that contains $C\setminus \{v\}$. Consider $T$ as a rooted tree with $B_r$ as root. Let $T_v$ be the subtree of bags that contain $v$, and let $B_v$ be the lowest common ancestor of all bags in $T_v$. Because $T_v$ is connected, $v \in B_v$.

    For any $w \in C \setminus \{v\}$, $(v,w)$ is an edge so there must be a bag $B_w \in T_v$ that contains both $v$ and $w$. Now, $w$ is contained in both a descendant ($B_w$) and ancestor ($B_r$) of $B_v$, so $w \in B_v$. Since this holds for any $w \in C \setminus \{v\}$, we conclude that $C \subset B_v$ which completes the induction. 
    \qed
\end{proof}

Some notation: for a graph $G = (V,E)$, tree decomposition $T$, and vertex set $C \subset V$, we define $T_C$ as the subset of bags in $T$ that contain $C$.
\begin{lemma}\label{lemma:treedecomp}
    Let $T$ be a tree where each vertex (bag) is associated with a subset of vertices from $G$. Then $T$ is a tree decomposition of $G$ if and only if: for each clique $C$ in $G$, $T_C$ is nonempty and connected.
\end{lemma}
\begin{proof}
    ($\Rightarrow$): Let $T$ be a tree satisfying this condition. All three conditions to being a tree decomposition follow directly from the fact that single vertices and edges are cliques.

    ($\Leftarrow$): Let $T$ be a tree decomposition. Let $C$ be any clique. \autoref{lemma:full} shows that $T_C$ is nonempty. We now show that it is connected. Let $B_1$, $B_2$ be two bags that fully contain $C$. For any vertex $v \in C$, we know that $T_{\{v\}}$ is connected, so every bag on the path between $B_1$ and $B_2$ contains $v$. Since this holds for any $v \in C$, the bags between $B_1$ and $B_2$ fully contain $C$. It follows that $T_C$ contains every bag between $B_1$ and $B_2$. Since this holds for any $B_1$ and $B_2$, we conclude that $T_C$ is connected.
    \qed
\end{proof}

We now prove the equivalence.
\begin{theorem}\label{thm:TCG-TMG}
    The problems \TCG{} and \TMG{} are equivalent under parameterized reductions. The parameter does not change under these reductions.
\end{theorem}
\begin{proof}
    The right implication is trivial: each instance of TCG is also an instance of TMG and each corresponding solution to TMG is also a solution to the TCG instance. We focus on the left implication. Let $G = (V,E)$ be an instance of TMG (with $k$ colors). Construct a graph $G'$ as follows. For each vertex $v \in V$ (with $k_v$ colors) we create a clique $C_v$ containing $k_v$ vertices, each colored in one of the colors of $v$. For each edge $(v,w) \in E$ we add an edge between each pair of vertices from $C_v$ and $C_w$ to turn $C_v + C_w$ into a large clique. The resulting graph $G' = (V',E')$ is now an instance of TCG. We claim that it has a solution if and only if the original TMG instance has a solution.

    \subsubsection{First Direction.}
    Let $T$ be a tree decomposition of $G$ that respects the multicoloring. For each vertex $v$, we replace each occurrence of $v$ in bags of $T$ with all the vertices from $C_v$. Call the result $T'$. We first show that $T'$ is a tree decomposition of $G'$ using \autoref{lemma:treedecomp}.

    Let a clique $C' \subset V'$ be given. For any vertex $v \in V$, we know that a bag of $T'$ contains either all vertices from $C_v$ or none. This implies the following: if $C'$ contains some (but not all) vertices from $C_v$, then the set of bags that contain $C'$ is the same as the set of bags that contain $C' \cup C_v$. Hence, without loss of generality we may assume that $C'$ contains either all vertices from $C_v$ or none.

    Now, let $C$ be the set of vertices $v \in V$ such that $C_v \subset C'$. By construction of $T'$, we have that any bag in $T$ contains $C$ if and only if the corresponding bag in $T'$ contains $C'$. Since $T_C$ is a nonempty subtree, we find that $T'_{C'}$ must also be a nonempty subtree. As this holds for any $C'$, we conclude that $T'$ is a tree decomposition.

    Additionally, $T'$ respects the coloring: each vertex from $V$ splits into a clique $C_v \subset V'$ which contains the same colors. Hence, each bag in $T$ contains precisely the same colors as the corresponding bag in $T'$ which implies that $T'$ cannot contain a color more than once. This completes the first direction.

    \subsubsection{Second Direction.}
    Let $T'$ be a tree decomposition of $G'$ that respects the coloring. We now construct the tree $T$ as follows: $T$ has the same shape as $T'$, and a bag in $T$ will contain a vertex $v$ if and only if the corresponding bag in $T'$ contains every vertex from $C_v$. We first show that $T$ is a tree decomposition.

    Let a clique $C \subset V$ be given. A bag in $T$ now contains $C$ if and only if the corresponding bag in $T'$ contains $C_v$ for each vertex $v \in C$. This set $C' := \bigcup_{v \in C} C_v$ is a clique: each $C_v$ is a clique itself and for each two cliques $C_v, C_w$ we have that $v$ and $w$ are connected in $G$ (since $C$ is a clique) and consequently that every vertex from $C_v$ is connected to every vertex from $C_w$. Now, the alternate definition of tree decomposition shows that $T'_{C'}$ is a nonempty subtree. This implies that $T_C$ is one as well. Since this holds for any clique $C$, we conclude that $T$ is a tree decomposition.

    Additionally, $T$ respects the coloring: each bag in $T$ corresponds to a bag in $T'$ that contains the same colors (and possibly some more). Since $T'$ respects the coloring, $T$ must do so as well. This completes the second direction.

    We have now shown that the constructed instance of TCG is equivalent to the original instance of TMG. Each vertex in $G$ splits into at most $k$ vertices in $G'$, and each edge splits into at most $k^2$ edges in $G'$. Hence, $G'$ contains at most $nk$ vertices, $mk^2$ edges and $k$ colors. Since this is polynomially many more than the original instance and since the parameter $k$ did not change, the reduction is polynomial. This completes the proof.
    \qed
\end{proof}


\section{Full Proof of XALP-hardness}\label{sec:fullproof}

In this section we will show that the construction given in \autoref{sec:hardness} admits a triangulation if and only if the original TCMIS instance admits a solution. This completes the proof of \autoref{thm:xalphard}.

\subsubsection{First Direction.}
We will first show that if the TMG instance can be triangulated, then there is a solution to the TCMIS instance. Suppose that the TMG instance admits a triangulation. Because of \autoref{prop:triang}\ref{prop:triang:treedecomp} this triangulation corresponds with a tree decomposition.

We first introduce some lemmas on tree decompositions.
\begin{lemma}[{\cite[Proposition~4]{BFW_92}}]\label{lemma:bagpath}
    Let $G$ be a graph and $P$ a path between two vertices $v,w$. In any tree decomposition $T$ of $G$, let $B_1$ and $B_2$ be two bags such that $v \in B_1$ and $w \in B_2$. Now, each bag on the path from $B_1$ from $B_2$ in $T$ contains at least one vertex from $P$.
\end{lemma}
\begin{lemma}\label{lemma:vertexpath}
    Let $G$ be a multicolored graph and $P$ a path between two vertices $v,w$ whose vertices alternate between two colors. In any tree decomposition $T$ of $G$, let $B_1$ and $B_2$ be two bags such that $v \in B_1$ and $w \in B_2$. Now, each vertex from $P$ is contained in at least one bag on the path from $B_1$ from $B_2$ in $T$.
\end{lemma}
\begin{proof}
    Omitted. See \autoref{sec:proofs}.
    \qed
\end{proof}

Because of \autoref{prop:zipper-gadget}, each zipper gadget $z_{n,c}$ is triangulated with some offset $i_{n,c}$ between 0 and $r$. We now claim that, for each node $n$ and color $c$, choosing the $i_{n,c}$-th vertex of $S_c$ forms a solution to the TCMIS instance. To show this, we only need to show that we have chosen at most one endpoint for each edge from the TCMIS instance. Let $e_j$ be an edge with endpoints $v_{n_1,c_1,i_1}$ and $v_{n_2,c_2,i_2}$, let $(P_1,Q_1)$ and $(P_2,Q_2)$ be the paths of the corresponding zipper gadgets $z_{n_1,c_1}$ and $z_{n_2,c_2}$, and suppose to the contrary that $z_{n_1,c_1}$ and $z_{n_1,c_1}$ are triangulated with offsets $i_1$ and $i_2$ (respectively). We consider two cases, either $n_1 = n_2$ or they are neighbors in $T$.

In the first case, the head, the middle and the first vertex of tooth $jr$ of the paths $P_1$ and $P_2$ are pairwise identified with each other. We label these vertices as $h$, $m$, and $u$ respectively. Let $B_h$ and $B_m$ be (any) bags that contain $h$ and $m$ (respectively). Note that there are four paths on $G$ between $h$ and $m$: each zipper gadget introduces two paths. Also, each of these paths alternates in color. Because of \autoref{lemma:vertexpath}, there must be a bag $B'$ on the path between $B_h$ and $B_m$ that contains $u$. Because of \autoref{lemma:bagpath}, $B'$ must contain a vertex from both $Q_1$ and $Q_2$, say $v_1$ and $v_2$.

We now claim that $v_1$ and $v_2$ both contain the color $d$. First consider $v_1$, the other case is analogous. Since $P_1$ was triangulated with offset $i_1$, the vertex $u$ must be connected to some vertex from tooth $ir+i_1$ from $Q_1$. Since $u$ has colors $a$ and $c_1$, the vertex from $Q_1$ can not have those colors. That leaves only one option: the first vertex (with colors $b_Q$ and $c_3$). This is precisely the vertex that was given the color $d$ in the construction, which completes the claim.

We now have a contradiction: $B'$ contains two vertices with color $d$ (one from each zipper gadget). This shows that we cannot have chosen both endpoints of $e_j$ and completes the first case.

The second case is almost analogous. We assume without loss of generality that $n_1$ is a parent of $n_2$. We now have that the head, middle, and first vertex of tooth $ir+i_1$ from $z_{n_1,c_1}$ are (respectively) identified with the middle, tail, and first vertex of tooth $mr+1+ir+i_2$ from $z_{n_2,c_2}$. We now apply the same reasoning as in the first case on these vertices. Overall, this completes the first direction.

\subsubsection{Second Direction.}
We will now show that if the TCMIS instance has a solution, then the TMG instance can be triangulated. For each node $n$ and color $c$, the TCMIS solution consists of a choice of some vertex $i_{n,c}$. In the TMG instance, this will correspond to a triangulation of the zipper gadget $z_{n,c}$ with offset $i_{n,c}$. This alone is not enough to obtain a triangulation; we need to add more edges. We will do this by creating a colored tree decomposition, which corresponds to a triangulation because of \autoref{prop:triang}\ref{prop:triang:treedecomp}.

We describe the tree decomposition in three steps. In the first step, we add a bag $B_n$ for each node $n$ of $T'$. This bag contains the shared vertex from all zipper gadgets that start, pass through, or end in $n$. It also contains a vertex of the path $Q$ from each of those zipper gadgets: the first vertex for each zipper gadget that starts in $n$, the last vertex for each zipper gadget that ends in $n$, and the first vertex of tooth $rm+1+\Delta$ (where $\Delta$ is the offset of this zipper gadget) for each zipper gadget that passes through $n$.

In the second step we will add $m-1$ bags between each pair of bags from the previous step. Let $n_1$ and $n_2$ be two adjacent nodes from $T'$ and let $B_1$ and $B_2$ be the corresponding bags from the previous step. Without loss of generality, we assume that $n_2$ is the parent of $n_1$. We now create a sequence of bags $B'_0, B'_1, \ldots, B'_m$ with $B'_0 := B_1$ and $B'_m := B_2$. The path between $B'_{i-1}$ and $B'_i$ will correspond to the edge $i$.

Each bag $B'_i$ will contain the first vertex from the $ir$-th tooth of the path $P$ for every zipper gadget that starts at $B_1$ and passes through $B_2$. Additionally, it will contain the first vertex of the $ir+\Delta$-th tooth of the path $Q$ for those zipper gadgets. Similarly, for zipper gadgets that pass through $B_1$ and end at $B_2$, it will contain the first vertex from the $mr+1+ir$-th tooth of $P$ and the first vertex from the $mr+1+ir+\Delta$-th tooth of $Q$. This completes the second step.

Note that all the bags we have added so far do not contain any color more than once: each two zipper gadgets that share a common start, end, or middle vertex have distinct colors. Furthermore, each bag contains the extra color $d$ at most once: each bag $B'_i$ only contains the color $d$ if $n_1$ contains an endpoint of edge $i$ and if the zipper gadget corresponding to the color of that endpoint was triangulated with the proper offset. That is, if we chose the endpoint. Since the TCMIS instance chooses at most one endpoint for each edge, $B'_i$ contains the extra color $d$ at most once.

Recall that in the construction of the TMG instance, each edge $i$ resulted in merging a vertex from two paths $P$ for the zipper gadgets that correspond to its endpoints. In the bags we have added so far, this merged vertex does not cause a problem: it ends up in the same bag when viewed from both zipper gadgets.

In the third step, we add more bags between two adjacent bags from the previous step. Let $B'_i$ and $B'_{i+1}$ be two adjacent bags that correspond to the edge $i$. Recall that these bags each contain the first vertex of some tooth for some set of zipper gadgets. Also, the vertices between them do not contain any merged vertices; those where all handled in the previous steps.

For each zipper gadget on its own, we could easily complete the tree decomposition between these bags using the tree decompositions that correspond to the triangulations of the zipper gadgets (note that these are paths). For the set of zipper gadgets as a whole, we need some care to avoid adding the color $d$ to the same bag twice. We do this by ``sliding along'' the zipper gadgets one by one: we start with some bags where we partially follow the triangulation of one zipper gadget, then add some bags where follow the triangulation of a second zipper gadget, and so on. We may also slide a bit further along the same zipper gadget multiple times. The exact order is described below.

\begin{itemize}
    \item First, if $B'_i$ contains the color $d$ because of some zipper gadget, slide along that zipper gadget for 1 tooth. This way, the current bag no longer contains the extra color $d$.
    \item Now, loop over every zipper gadget one by one except possibly the one that causes $B'_{i+1}$ to contain the color $d$. For each such zipper gadget, slide along it until we arrive at the tooth contained in bag $B'_{i+1}$. During this, there may have been some bags which contain the color $d$ but it will not be in the current bag when we continue to the next zipper gadget.
    \item Finally, handle the zipper gadget that causes $B'_{i+1}$ to contain the color $d$ (if it exists). We slide over it until we reach the vertex contained in $B'_{i+1}$. Here, the final bag will contain the color $d$.
\end{itemize}
This completes the third step and thus also the construction of the tree decomposition. Because of the arguments given during the construction, each bag contains each color at most once. Since all bags consist of an interlocking of sliding along the zipper gadgets, we have that each vertex is contained in a connected subtree and that the endpoints of each edge are contained in some bag. This proves that the above construction is indeed a tree decomposition. We conclude that a tree decomposition exists, and consequently that the TMG instance has a solution. This completes the second direction.

Overall, we have now shown that the original TCMIS instance admits a solution if and only if the constructed TMG instance does. This completes the proof of \autoref{thm:xalphard}.
\qed


\section{Remaining Proofs}\label{sec:proofs}
This section includes some proofs that were omitted from the main text.

\begin{proof}[of \autoref{prop:triang}]
    We prove the parts in order.
    \begin{enumerate}[label=(\roman*)]
        \item Consider a triangulation of $G$. We use induction on the size of $C$. If $\abs{C} = 3$, then one of the two colors must occur at least twice, hence one of the edges from $C$ connects two vertices sharing a color. If $\abs{C} \geq 4$, then $C$ contains a chord. This chord splits $C$ into two smaller cycles which share the chord as a common edge. Applying the induction hypothesis on one of these cycles shows that $G$ cannot be triangulated into a properly colored graph.
        \item Consider a triangulation of $G$ and suppose to the contrary that both the edge and the chord do not exist. We use induction on the size of $C$. If $\abs{C} = 3$, then $v$'s neighbors are connected. If $\abs{C} \geq 4$, then $C$ contains a chord. This chord cannot connect $v$'s neighbors or have $v$ as an endpoint. Hence, it splits $C$ into two smaller cycles (which share the chord as a common edge) such that both $v$ and its neighbors are contained in one of them. Applying the induction hypothesis on this cycle completes the proof.
        \item For colored graphs, the left implication was shown in~\cite{B_74} and the right implication in~\cite{G_74}. The proof for multicolored graphs is very similar to the colored version and thus omitted.
    \end{enumerate}
    \qed
\end{proof}

\begin{proof}[of \autoref{lemma:exact}]
    The right implication is trivial. We focus on the left implication. Let $T$ be a tree decomposition. If there is a color that is contained in some but not all bags, then there must exist two adjacent bags $B_1$, $B_2$ such that $B_1$ does not contain a vertex with this color and $B_2$ does. We can then add this vertex from $B_2$ to $B_1$ and observe that the result is still a valid tree decomposition where each bag contains each color at most once. By repeating this argument, we must eventually reach a state where every bag contains all colors.
    \qed
\end{proof}

\begin{proof}[of \autoref{lemma:vertexpath}]
    Let $u$ be a vertex from $P$ and suppose to the contrary that no bag on the path from $B_1$ and $B_2$ contains $u$. We split $P$ into two parts: $P_1$ from $v$ to $u$ and $P_2$ from $u$ to $w$. Let $B_3$ be (any) bag that does contain $u$. Because of \autoref{lemma:bagpath}, each bag on the path from $B_1$ to $B_3$ contains some vertex from $P_1$. Analogously, each bag on the path from $B_3$ to $B_2$ contains some vertex from $P_2$. Since $T$ is acyclic, the paths between $B_1$, $B_2$ and $B_3$ must intersect in some point, so there is a bag $B$ which lies on all three paths. This bag then contains a vertex from $P_1$, a vertex from $P_2$, and it does not contain $u$. In particular, it contains two non-adjacent vertices from $P$. That means that, in the corresponding triangulation of $G$, there is an edge between two non-adjacent vertices of $P$. This edge combined with $p$ induces a cycle which alternates between two colors and that is impossible because of \autoref{prop:triang}\ref{prop:triang:alternate}. Hence, we arrive at a contradiction.
    \qed
\end{proof}

\end{document}